\documentclass[10pt, conference, letter]{IEEEtran}

\usepackage{cite}
\usepackage[pdftex]{graphicx}
\graphicspath{{./figs/}}
\DeclareGraphicsExtensions{.pdf,.jpeg,.png,.eps}
\usepackage{amsmath}
\usepackage{amssymb}
\usepackage{latexsym}
\usepackage{array}
\usepackage{algorithmic}
\usepackage{algorithm}
\usepackage{float}

\usepackage{color}
\usepackage{bm}
\usepackage{setspace}
\usepackage{fancyhdr}
\newcommand{\subparagraph}{}
\usepackage{titlesec}
\usepackage[stable]{footmisc}
\usepackage[table,xcdraw]{xcolor}
\usepackage{csquotes}
\MakeOuterQuote{"}
\makeatletter
\newcommand{\vast}{\bBigg@{2.6}}
\newcommand{\Vast}{\bBigg@{4.6}}
\newcommand{\ignore}[1]{}
\makeatother

\newtheorem{theorem}{Theorem}

\usepackage{parskip}
\usepackage{hyperref}

\usepackage{balance}
\begin{document}
	\title{Jointly Optimizing Dataset Size and Local Updates in Heterogeneous Mobile Edge Learning}
	
	\author{
		\IEEEauthorblockN{Umair Mohammad\IEEEauthorrefmark{1}, Sameh Sorour\IEEEauthorrefmark{2}, Mohamed Hefeida\IEEEauthorrefmark{3}}
		\IEEEauthorblockA{\IEEEauthorrefmark{1}School of Computing and Information Sciences, Florida International University, Miami, FL, USA}
		\IEEEauthorblockA{\IEEEauthorrefmark{2}School of Computing, Queen’s University, Kingston, ON, Canada}\IEEEauthorblockA{\IEEEauthorrefmark{3}Department Computer Science and Electrical Engineering, West Virginia University, Morgantown, WV, USA}
		E-mails: umair.mohammad@fiu.edu, sameh.sorour@queensu.ca, mohamed.hefeida@mail.wvu.edu
	}
	
	\maketitle
	
	\begin{abstract}
		This paper proposes to maximize the accuracy of a distributed machine learning (ML) model trained on learners connected via the resource-constrained wireless edge. We jointly optimize the number of local/global updates and the task size allocation to minimize the loss while taking into account heterogeneous communication and computation capabilities of each learner. By leveraging existing bounds on the difference between the training loss at any given iteration and the theoretically optimal loss, we derive an expression for the objective function in terms of the number of local updates. The resulting convex program is solved to obtain the optimal number of local updates which is used to obtain the total updates and batch sizes for each learner. The merits of the proposed solution, which is heterogeneity aware (HA), are exhibited by comparing its performance to the heterogeneity unaware (HU) approach.
	\end{abstract}
	
	\section{Introduction}
	\label{Section1_Introduction}
	Rapid migration towards smart infrastructure (cities, cars, grids, etc.) has caused an explosion of Internet-of-Things (IoT) devices on resource-constrained wireless edge networks. 
	A recent article reported reports that about every second another 127 devices connect to the internet with 41 billion devices expected by 2027 and Cisco expects that 800 zettabytes of data will be generated \cite{N202003_vXchnge_IoT_data} on wireless edge networks. The distributed nature of this data will place a heavy financial burden on backbone networks and raise security/privacy concerns \cite{Ref1_Original}. 
	Thus, it is anticipated that edge servers and end devices (e.g. smart phones, cameras, drones, connected vehicles, etc.) will perform 90\% of the data processing locally \cite{Ref2_Original}.  
	
	Machine Learning (ML) techniques have shown to perform better in many data analytics applications such as forecasting, image classification, clustering, etc. Many ML techniques, including regression, support vector machine (SVM) and neural networks (NN) are built on gradient-based learning. 
	This usually involves optimizing the model parameters by iteratively adding to them the gradient of the loss, itself a function of the model parameters.
	In the distributed learning model considered in this paper, a central server called an orchestrator initiates the learning process on multiple learners where each learner performs the ML iterations on the local dataset, collects the local ML models from each learner, does the global update/aggregation, and sends back the optimal ML model for the next cycle until a stopping criteria is reached.
	
	With the advent of Edge Artificial Intelligence (Edge AI), deploying ML models over end devices at edge networks will soon be the norm. Therefore, researchers have turned their focus to performing machine learning (ML) in a distributed manner (a.k.a. distributed learning (DL)) at the edge \cite{HarvardDistributed, Tuor_01_AdaptiveControl, Wang2019, Mohammad2019a, Tuor04, J1902_ResAllocMobDL_TPDS_WuZhang} in order to support edge analytics. 
	In general, DL at the edge can be characterized by mobile edge learning (MEL) though the most commonly studied setup is federated learning (FL).
	
	The works of \cite{Tuor_01_AdaptiveControl, Wang2019, Tuor04, J1902_ResAllocMobDL_TPDS_WuZhang} focus on jointly optimizing the number of local learning and global update cycles for FL. However, their approaches do not consider the inherent heterogeneity in the computing and communication capacities of different edge learners and links, respectively.
	Although the works of \cite{2019arXiv190907972C,2019arXiv191102417Y} have optimized resource allocation while maintaining accuracy, they do not investigate the impact of batch allocation. 
	The implications of wireless computation/communication heterogeneity on optimizing batch allocation to different learners for maximizing accuracy while satisfying a delay constraint were studied in \cite{Mohammad2019a}.
	
	To the best of the authors' knowledge, this work is the first attempt at jointly optimizing batch size allocated to learners and synchronizing the number of local and total iterations of the ML algorithm across all learners while satisfying time delay constraints for the MEL paradigm. Therefore, our proposed work differs from existing literature in the following ways: 1) considers the impact of batch size allocation 2) models the MEL system on actual channel parameters and device capabilities as opposed to generic resource consumption 3) optimizes both, the number of local and total updates as opposed to maximizing the local updates per global update.
	
	The superiority of our proposed heterogeneity aware (HA) approach is shown by comparing its performance to the heterogeneity-unaware (HU) approach of \cite{Tuor_01_AdaptiveControl,Wang2019}. 
	Tests on classifying the MNIST dataset \cite{MNIST_IEEE} using a DNN show that the HA approach is superior in terms of achieving a lower loss and providing higher validation accuracy. The rest of the paper is organized as follows: Section 2 introduces the global MEL model with time constraints. The problem of interest in this paper is formulated in Section 3 and our proposed solution is described in Section 4. Section 5 presents the results and Section 6 concludes the paper. 
	
	\section{MEL System Model}
	\label{Section02__SystemModelParameters}
	
	\subsection{Gradient-based Learning Preliminaries}
	
	Consider a dataset that consists of $d$ samples that can be trained using ML where each sample $n$ for $n=1,\ldots,d$ has a set of $\mathcal{F}$ features denoted by $\mathbf{x}_n$ and a target $y_n$. The objective is to find the relationship between $\mathbf{x}_n$ and $y_n$ using a set of parameters $\mathbf{w}$ such that a loss function, $F\left(\mathbf{x}_n, \mathbf{y}_n, \mathbf{w}\right)$ (or $F\left(\mathbf{w}\right)$ for short), is minimized. Because it is generally difficult to find an analytical solution, typically an iterative gradient descent approach is used to optimize the set of model parameters such that $\mathbf{w}[l+1] = \mathbf{w}[l] - \eta\nabla F\left(\mathbf{w}[l]\right)$ where $l$ represents the time step or iteration and $\eta$ is the learning rate typically set on the interval $(0,1)$. In deterministic gradient descent (DGD), the ML model goes over each sample one-by-one, or more commonly, batch-by-batch using a mini-batch approach, until it reaches sample \# $d$; completing one epoch. If data is re-shuffled randomly in between epochs, this method is know as stochastic GD (SGD). A total of $L$ epochs may be performed depending on the stopping criteria.

	\subsection{Transition to MEL}
	An MEL system consists of an orchestrator and $K$ learners where $d_k$ data samples are allocated to learner $k$, $k \in \mathcal{K} = \{ 1, 2, \dots, K\}$ so that it performs $\tau$ learning iterations. Each learner has a computational capacity of $f_k$ in Hz and an associated communication channel $h_{k0}$ to the orchestrator. An example of a DL system is illustrated in fig. \ref{figure0label}. We assume that $K$ perfectly orthogonal channels exist.
	\begin{figure}[t!]
		\centering
		\includegraphics[width=\linewidth,keepaspectratio]{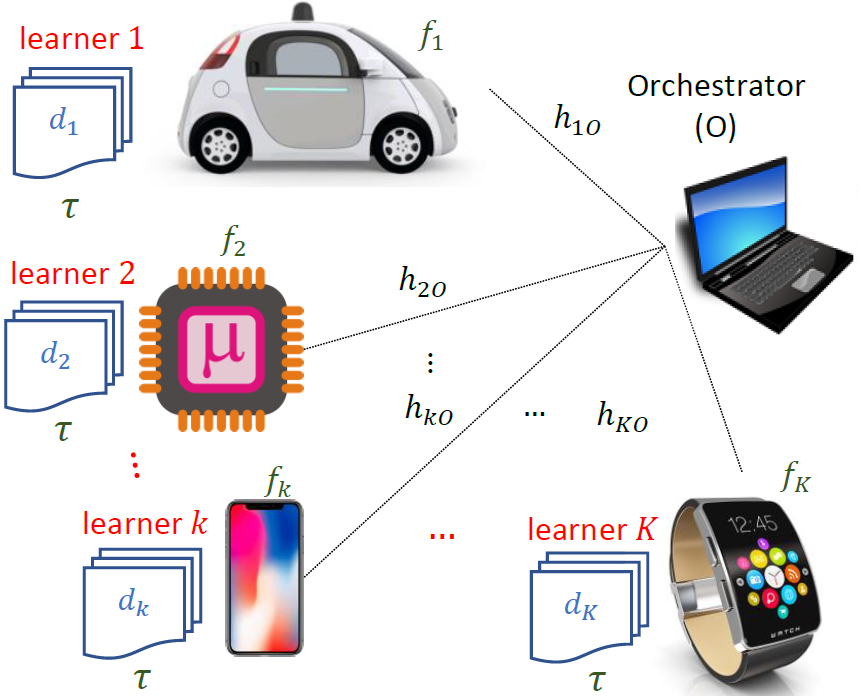}
		\caption{System model of a MEL setting}
		\label{figure0label}
	\end{figure}
	
	DL as described in section \ref{Section1_Introduction} in an MEL setting gives rise to two possibilities: offloaded learning (OL) and federated learning (FL). In the former, the orchestrator has the complete dataset and also re-transmits optimally allocated batches from the randomly shuffled dataset back to the learners. In the latter, the orchestrator only informs the learner of how many iterations to do and on what sample size of a locally stored dataset. This approach has been more commonly studied in literature as opposed to OL. FL is just a subset of the OL approach with the component of batch re-transmission from the orchestrator to each learner $k$ removed. Thus, the MEL model discussion will focus on the more general offloaded learning scenario but all variations for the federated learning scenario will be clarified whenever needed. 	
	
	We define $B_k^{data} = d_k\mathcal{F}\mathcal{P}_d$ as the size of the batch allocated to learner $k$ in bits and $B_k^{model} = \mathcal{P}_m\left(d_kS_d+S_m\right)$ as the size of the model in bits $~ \forall k \in \mathcal{K}$. The variables $\mathcal{P}_d$ and $\mathcal{P}_m$ represent precision with which the data and model are stored, respectively, $\mathcal{F}$ represents the feature vector size of $\mathbf{x}_n$ for $n=1,\dots, d$. $\mathcal{S}_m$ represents size of $\mathbf{w}$ as defined by the ML model whereas $S_d$ is the proportion of the model dependent on the dataset size.
	
	Between any two global update cycles, the orchestrator, which owns the global model, sends the data and model parameters $\mathbf{w}$ to each learner $k$ in parallel\footnote{In the FL scenario, the orchestrator only sends the global model $\mathbf{w}$ to each learner $k \in \mathcal{K}$. Learner $k$ selects $d_k$ samples from its private dataset.}, waits for all learners to complete $\tau$ local learning iterations, and then receives the locally updated model $\mathbf{w}_k \forall k \in \mathcal{K}$ followed by global aggregation. The communication of the models (and data) between the orchestrator and each learner $k$ occurs over a channel having a bandwidth $W$, a  channel power gain $h_{ko}$ and with a noise power spectral density of $N_0$. We assume that $P_{ko} = P_{ok}$ over one iteration of the global update and that channel parameters remain constant during global aggregation. 
	
	Furthermore, learner $k ~ \forall ~ k \in \mathcal{K}$ has a local processor resource $f_k$ dedicated to the DL task for an ML model of complexity $C_m$ that requires $X_k$ clock cycles to perform one local iteration. Given the above descriptions, the times of each learner $k$, $\forall~k$ comprise: orchestrator transmission time $t_k^S$ needed to send $\mathbf{w}$ and $d_k$ data samples\footnote{For the FL scenario, the only difference in the model is that the first term of the numerator will not exist.} to learner $k$, the duration $t_k^C$ needed by learner $k$ to perform one local update cycle, and time $t_k^R$ is the one needed for learner $k$ to send its updated local parameter matrix $\tilde{\mathbf{w}}_k$ to the orchestrator. The times $t_k^S$,$t_k^C$, and $t_k^R$, respectively, can be expressed as:
	\begin{equation}
	t_{k}^S = \dfrac{d_k\mathcal{F}\mathcal{P}_d + \mathcal{P}_m \left(d_k\mathcal{S}_d+\mathcal{S}_m\right)}{W\log_2\left(1+\frac{P_{ko} h_{ko}}{N_0}\right)}
	\label{eq:edge_time_sending}
	\end{equation}
	\begin{equation}
	t_k^C = \dfrac{X_k}{f_k} = \frac{d_k C_m}{f_k}
	\label{Eq_9_localTimeExecution}
	\end{equation} 
	\begin{equation}
	t_{k}^R = \dfrac{\mathcal{P}_m \left(d_k\mathcal{S}_d+\mathcal{S}_m\right)}{W\log_2\left(1+\frac{P_{ko} h_{ko}}{N_0}\right)}
	\label{eq:edge_time_receiving}
	\end{equation}
	
	\section{Problem Formulation}
	\label{Section03_ProblemFormulation}	
	
	As mentioned in Section I, the objective of this paper is to optimize the task allocation, i.e. the distributed batch sizes $d_k$ for each learner $k$ and the associated $\tau$ updates to be performed locally for a total of $L$ updates, such that the global DL loss is minimized and thus, the accuracy is maximized. To this end, the problem is formulated as a loss-function minimization problem over the optimization variables $L$, $\tau$ and $d_k$. 
	
	Consider that after every $\tau$ local iterations, a global aggregation will be performed and a total of $G$ global aggregations are performed. In any global update cycle, between any learner $k \in \mathcal{K}$ the orchestrator $O$, there will be one communication round each and $\tau$ local updates. For now, to facilitate the analysis, let us assume that $L$ is an integer multiple of $\tau$ such that $L=G\tau$ and that the communication and computation related parameters remain unchanged over the complete training process. In that case, each learner needs time $t_k^C ~ \forall ~ k \in \mathcal{K}$ for one local update and $t_k^S+t_k^R$ for the $g^{th}$ global aggregation for $g=1,\dots,G$. Overall, $L$ local updates and $G=L/\tau$ global updates will be performed. Then,the total time consumed by learner $k$ denoted by $t_k ~ \forall k \in \mathcal{K}$ can be expressed as:  
	\begin{equation}
	t_k = L\left(t_k^C + \dfrac{t_k^S + t_k^R}{\tau}\right)
	\end{equation}
	
	Later on, we will shoe how the values of $\tau$ and $d_k$ for each set of $\tau$ local ML iterations in one global cycle $g$ will be re-calculated according to the latest channel parameters and computational capabilities. The total training time within which the process should be completed is given bounded by $T$. Because the $\tau$ iterations occur in parallel over the $K$ learners, we need the time for the most time-consuming learner to be less than $T$ such that $\max(t_k) \leq T$. Alternatively, it is sufficient for this condition to hold that $t_k \leq T  ~ \forall k \in \mathcal{K}$. 
	
	This point differentiates our work from that of \cite{Tuor_01_AdaptiveControl} where we actually capture the time consumed by parallel local update processes rather than a generic resource consumption model. 
	Therefore, the optimization problem can be written as:
	\begin{subequations}
		\begin{align}
		&\qquad\operatornamewithlimits{min}_{L,\tau,{d}_k\forall~k}  \quad F(\mathbf{w}[L]) \tag{\ref{eq_Original_Prob}}\\
		& \quad \nonumber\\
		\text{s.t. }\qquad & L\left(C_k^2  d_k + \dfrac{C_k^1 d_k + C_k^0}{\tau}\right) \leq T, \quad k \in \mathcal{K} \nonumber
		\end{align}
		\label{eq_Original_Prob}
	\end{subequations}
	
	The constants $C_k^2$,$C_k^1$,and $C_k^0$ can be defined as:
	\begin{subequations}
	\begin{align}
	C_k^2 = \frac{\mathcal{C}_m}{f_k} 
	\end{align}
	\begin{align}
	C_k^1 = \frac{\mathcal{F}\mathcal{P}_d+2\mathcal{P}_m\mathcal{S}_d}{W\log_2\left(1+\frac{P_{ko} h_{ko}}{N_0}\right)} 
	\end{align}
	\begin{align}
	C_k^0 = \frac{2\mathcal{P}_m\mathcal{S}_m}{W\log_2\left(1+\frac{P_{ko} h_{ko}}{N_0}\right)} 
	\end{align}
	\end{subequations}
	It is generally impossible to find an exact expression relating the optimization variables to the objective for most ML models. Therefore, the objective will be re-formulated as a function of the convergence bounds on the DL process over the edge. For more details on these bounds, the readers are referred to \cite{Tuor_01_AdaptiveControl}. We will use the results and extend the discussion to our formulation, and then propose a strategy to jointly find the optimal $\tau$, $d_k$, and $L$.   
	
	\subsection{Convergence Bounds} 
	The convergence bounds have been derived and well-discussed in \cite{Tuor_01_AdaptiveControl}. For completeness, we will present some of the important results here in order to support our analysis. Let us continue with the assumption that $L$ is an integer multiple of $\tau$. Then, the global aggregation will only occur at every $\tau$ updates. I.e. the local updates occur at every iteration $l=1,\ldots,L$ and a global update will occur whenever $ l = g\tau$ for $g=1,\ldots,G$. For any interval $[g]$ defined over $[(g-1)\tau,g\tau]$, define an auxiliary global model denoted by $\mathbf{v}$ which would have been calculated if a global update occurred as follows:
	\begin{equation}
	\mathbf{v}_{[g]}[l] = \mathbf{v}_{[g]}[l-1]-\eta \nabla F(\mathbf{v}_{[g]}[l-1]) 
	\end{equation} 
	
	Let the local model parameter set of learner $k$ be denoted by $\mathbf{w}_k$ and the local loss by $F_k(\mathbf{w}_k)$. Then, the optimal model at iteration $l$ can be obtained by:
	\begin{equation}
	\mathbf{w}[l] = \dfrac{1}{d}\sum_{k = 1}^K d_k \mathbf{w}_k[l]
	\end{equation} 
	The optimal $\mathbf{w[l]}$ will only be visible when $l = g\tau$ and for that iteration, the global loss can be defined by:
	\begin{equation}
	F(\mathbf{w}) = \dfrac{1}{d}\sum_{k = 1}^K d_k F_k(\mathbf{w})
	\end{equation} 
	
	The following assumptions are made about the loss function $F_k(\mathbf{w})$ at learner $k$: $F_k(\mathbf{w})$ is convex, $\lVert F_k(\mathbf{w})-F_k(\bar{\mathbf{w}}) \rvert \leq \rho \lvert \mathbf{w}-\bar{\mathbf{w}} \rvert$, and $ \lVert \nabla F_k(\mathbf{w})-\nabla F_k(\bar{\mathbf{w}}) \rvert \leq \beta \lvert \mathbf{w}-\bar{\mathbf{w}} \rvert$ for any $\mathbf{w}$, $\bar{\mathbf{w}}$. These assumptions will hold for ML models with convex loss function such as linear regression and SVM. By simulations, we will show that the proposed solutions work for non-convex models such as the neural networks with ReLU activation. 
	
	Let us also assume that the local loss function at $F_k(\mathbf{w})$ does not diverge by more than $\delta_k$ such that $\lvert F_k(\mathbf{w})- F(\mathbf{w}) \rvert \leq \delta_k$ and $\delta = \frac{\sum_k d_k \delta_k}{d}$. Furthermore, $\lvert \mathbf{w} - \mathbf{v}_{[g]}[l-1] \rvert \leq h(l-(g-1)\tau)$. 
	For any $\tau$, $h(\tau) = \frac{\delta}{\beta}\left[\left(\eta\beta+1\right)^\tau-1\right]-\eta\delta\tau$. Recall that $\eta$ is the learning rate and $\beta$ can be estimated by $\beta = \dfrac{\sum_k d_k \beta_k}{d}$ where:
	\begin{equation} 
	\beta_k = \dfrac{\lVert \nabla F_k(\mathbf{w_k[l]})-\nabla F_k(\mathbf{w[l]}) \rvert}{\lvert \mathbf{w_k[l]}-\mathbf{w[l]}\rvert}
	\end{equation}
	
	Based on this, the objective can be written as a function of the difference of the global loss after iteration $L$ and the optimal global loss. Given the above assumptions about the loss function, and the constraints on the optimization variables and the time taken by learner $k\in\mathcal{K}$, the optimization problem can be written as:		
	\begin{subequations}
		\begin{align}
		&\qquad\operatornamewithlimits{min}_{L,\tau,{d}_k\forall~k}  \quad F(\mathbf{w}[L])-F(\mathbf{w}^*)\\
		& \quad \nonumber\\
		\text{s.t. }\qquad & L\left(C_k^2  d_k + \dfrac{C_k^1 d_k + C_k^0}{\tau}\right) \leq T, \quad k = 1,\ldots,K \label{bounded-time-const}\\
		& \sum_{k = 1}^{K}d_k = d \label{orignial-batch-const}\\ 
		& \tau \in \mathcal{Z}_+ \label{orignial-tau-const}\\
		& L \in \mathcal{Z}_+ \label{orignial-L-const}\\
		& d_k \in \mathcal{Z}_+, \quad k = 1,\ldots,K \label{orignial-d-const}\\
		& \eta(1-\frac{\beta\eta}{2})-\frac{\rho}{\omega\epsilon^2}\frac{h(\tau)}{\tau} \geq 0 \label{relaxed-minloss-const} \\
		& \eta \beta \leq 1 \label{relaxed-etabeta-const} \\
		& F\left(\mathbf{v}_[g][l]\right)-F(\mathbf{w}^*) \geq \epsilon \label{relaxed-localloss-const} \\
		& F(\mathbf{w}(L)-F(\mathbf{w}^*)) \geq \epsilon \label{relaxed-globalloss-const} 
		\end{align}
		\label{Eq_13_OurProb}
	\end{subequations}
	Constraint (\ref{bounded-time-const}) guarantees that the time consumed by a total of $L$ updates does not exceed the total training time available given by $T$ seconds. 
	Constraint (\ref{orignial-batch-const}) ensures that the total dataset comprising $d$ samples is utilized. Constraints (\ref{orignial-tau-const}) - (\ref{orignial-d-const}) are simply non-negativity and integer constraints for the optimization variables where $L$, $\tau$ and/or all $d_k$'s being zero represent cases where DL is not possible in the MEL environment. Constraints (\ref{relaxed-etabeta-const}) and (\ref{relaxed-minloss-const}) represent a bound on the learning rate, meaning it should be small enough such that it guarantees convergence. When (\ref{relaxed-etabeta-const}) holds, (\ref{relaxed-minloss-const}) will always hold. Constraints (\ref{relaxed-localloss-const}) and (\ref{relaxed-globalloss-const}) define a lower bound on the gap between the optimal loss and the auxiliary loss at interval $[g]$ and the global loss, respectively, where $\epsilon > 0$. The parameter $\omega \triangleq \min_g \lVert \mathbf{v}_{[g]}[g-1]\tau-\mathbf{w}^*\rVert^{-2}$ represents the interval that minimizes the difference between the auxiliary loss and the global loss. The variables $\rho$, $\omega$, and $\epsilon$ appear in a single term $\frac{\rho}{\omega\epsilon^2}$ in (\ref{relaxed-minloss-const}) which represents a control parameter. Later on, this term will be represented by $B_0$ but for now, we will continue with the original terms make the analysis relatable to the original variables.
	
	It is assumed that $\eta > 0$ (typically $0<\eta<1$) and $\beta > 0$. Furthermore $\eta$ and $\epsilon$ can be set to small enough values such that $\eta\beta \leq 1$, and the constraints in (\ref{relaxed-minloss-const})-(\ref{relaxed-globalloss-const}) are satisfied. For a $\beta$-smooth function, Bernoulli's inequality will hold implying that $(\eta\beta+1)^\tau \geq \eta\beta\tau+1$. Furthermore, once all the assumptions about the loss function constraints are satisfied, it can be shown that $F(\mathbf{w}[L])-F(\mathbf{w}^*) \leq \frac{1}{\eta(1-\frac{\beta\eta}{2})-\frac{\rho}{\omega\epsilon^2}\frac{h(\tau)}{\tau}}$. Thus, the problem in (\ref{Eq_13_OurProb}) can be re-formulated as:
	\begin{subequations}
		\begin{align}
		&\qquad\operatornamewithlimits{min}_{L,\tau,{d}_k~\forall~k}  \quad \dfrac{1}{L\left[\eta(1-\frac{\beta\eta}{2})-\frac{\rho}{\omega\epsilon^2}\frac{h(\tau)}{\tau}\right]}\\
		& \quad \nonumber \\
		\text{s.t. }\qquad & L\leq \dfrac{T\tau}{C_k^2\tau d_k + C_k^1 d_k + C_k^0}, \quad k = 1,\ldots,K \label{relaxed-time-const}\\
		& \sum_{k = 1}^{K}d_k = d \label{relaxed-batch-const}\\ 
		& \tau \in \mathcal{Z}_+ \label{relaxed-tau-const}\\
		& L \geq 0 \label{relaxed-L-const}  \\
		& d_k \geq 0, \quad k = 1,\ldots,K \label{relaxed-d-const} 	
		\end{align}
		\label{Eq_13_RelaxedProblem}
	\end{subequations}
	 Note that the integer constraints on $d_k$ and $L$ have been relaxed in (\ref{relaxed-d-const}) and (\ref{relaxed-L-const}) which will help in proposing a solution. 
	
	\section{Proposed Solution}
	\label{Section3_Solution}
	The idea of the proposed solution is to re-write the objective as a function of $\tau$ by using the constraints on the total time consumption and the fact that the system must train the model on at least $d$ training samples.    
	
	\subsection{Relating Bounds to $\tau$}
	The orchestrator can ensure that the bounds in constraints (\ref{relaxed-etabeta-const})-(\ref{relaxed-globalloss-const}) are satisfied by choosing small enough values for $\eta$ and $\epsilon$. 
	In that case, if constraint (\ref{relaxed-minloss-const}) holds, the denominator of the objective function will be positive. Furthermore, if we relax the integer constraint on $L$, the optimal value for the total learning iterations can be given by:
	\begin{equation}
	L = \dfrac{T\tau}{C_k^2\tau d_k + C_k^1 d_k + C_k^0}, \quad k = 1,\ldots,K \label{eq_LOptimal_dkAndTau}
	\end{equation}
	By using the equality constraint in (\ref{relaxed-batch-const}), and re-arranging (\ref{eq_LOptimal_dkAndTau}) to make $d_k$ the subject, and defining two new variables $a_k = \frac{C_k^1}{C_k^2}$ and $b_k = \frac{C_k^0}{C_k^2}$, we can write $L$ as a function of $\tau$.
	\begin{equation}
	L(\tau) = \dfrac{KT\sum_{k=1}^{K} \tau \prod_{\substack{l=1 \\ l\neq k}}^{K} \left(\tau^+b_l\right)}{d\prod_{k=1}^{K}(\tau + b_k)+\sum_{k=1}^{K} a_k \prod_{\substack{l=1 \\ l\neq k}}^{K} \left(\tau^+b_l\right)} \label{eq_LOptimal_Tau}
	\end{equation} 
	The objective function denoted by $O$ can be re-written as a function of $\tau$ in the following manner:
	\begin{equation}
	O(\tau) = \dfrac{1}{L(\tau)}\dfrac{1}{\left[\eta(1-\frac{\beta\eta}{2})-\frac{\rho}{\omega\epsilon^2}\frac{h(\tau)}{\tau}\right]}=\dfrac{P(\tau)}{L(\tau)}
	\label{eq_O_tau}
	\end{equation}
	\begin{theorem}
		$O(\tau)$ is strictly convex on the domain $\tau \geq 0$.
		\label{theorem1}
	\end{theorem}
	\begin{proof}
		Please refer to Appendix A for the proof.		
	\end{proof} 
	
		\begin{algorithm}[t]
		\caption{MEL Process at the Orchestrator}
		\label{alg1}
		\begin{algorithmic}[1]
			\renewcommand{\algorithmicrequire}{\textbf{Input:}}
			\renewcommand{\algorithmicensure}{\textbf{Output:}}
			\REQUIRE $T$, $B_0$, $\tau_{max}$, $d$, $K$
			\ENSURE  $\mathbf{w}[L]$
			\\Initialize $\tau \leftarrow 1$, $d_k \leftarrow \frac{d}{K}$, $L=1$, and $\hat{T} \leftarrow \max(t_k)$
			\\Set $\mathbf{w} \leftarrow \mathbf{w}[0]$ as a random vector
			\\ \textit{Parallel Process}
			\WHILE {$\tau > 0$}
			\STATE 
			Send $\mathbf{w}$ and $d_k$ samples to each learner $k$
			\STATE After $\tau$ local iterations, collect $\mathbf{w}_k$, $\beta_k$, and $\nabla F_k(\mathbf{w_k})$
			\STATE Estimate $\mathbf{w}$,$\nabla F(\mathbf{w})$, $\beta$, and $\delta$
			\STATE Receive $P_{kO}$, $h_{kO}$ and $f_k$ from each learner $k ~ \forall ~ k ~ \in ~ \mathcal{K}$
			\ignore{\STATE Determine $a_k$ and $b_k$ $~\forall ~ k ~ \in ~ \mathcal{K}$}
			\STATE Set $T \leftarrow T-\hat{T}$
			\STATE Find the optimal $\tau$ by using theorem 1
			\STATE Calculate $L$ using (\ref{eq_LOptimal_Tau}) 
			\STATE Use new value of $L$ to obtain $d_k ~ \forall ~ k$ using (\ref{eq_LOptimal_dkAndTau})
			\STATE Let $\hat{T} \leftarrow \hat{T} + \max(t_k)$
			\IF {$\hat{T} > T$}
			\STATE Reduce $\tau$ to maximum value $\geq 0$  such that $\hat{T} \leq T$
			\ENDIF
			\ENDWHILE
			\RETURN $\mathbf{w}$ 
		\end{algorithmic} 
	\end{algorithm} 
	
	Because $\frac{\partial O}{\partial \tau} = 0$ does not have a closed form solution, the optimal $\tau^*$ can be obtained by solving the following problem:
	\begin{equation}
	\tau* = \arg \min_{\tau} O(\tau)
	\end{equation}  
	The value of $\tau^*$ can be difficult to obtain because $\tau$ is unbounded. However, we can limit the search space by $\tau_{max}$ and then use a brute force approach to find the optimal $\tau^*$. In fact, a binary search procedure has been proposed in \cite{Tuor_01_AdaptiveControl} which has a complexity of $\mathcal{O}(\log \tau_{max})$. Once $\tau*$ has been determined, $L^*$ can be obtained using (\ref{eq_LOptimal_Tau}) and re-set to this new value. The values of $d_k^* ~ \forall ~ k \in \mathcal{K}$ for the next $\tau$ updates can be obtained using (\ref{eq_LOptimal_dkAndTau}). Because the integer constraint on $d_k ~ \forall ~ k \in \mathcal{K}$ was relaxed in (\ref{Eq_13_RelaxedProblem}), they can be set by flooring the actual value. This process is repeated for each global cycle until the total training time is consumed. This process is summarized in algorithm \ref{alg1}.	
	
	\section{Simulation Results}
	\subsection{Simulation Environment, Dataset, and Learning Model} 
	The learners are assumed to be located in a cellular type environment and are assumed to be a combination of smart phone and Raspberry PI type micro-controllers.
	The channel parameters and device capabilities are listed in table 1. To test our proposed MEL paradigm, the commonly used MNIST \cite{MNIST_IEEE} dataset is trained using a DNN with 3 hidden layers consisting of 300, 124 and 60 neurons, respectively. The details of the resulting model sizes and complexities are discussed in \cite{Mohammad2019a}.
	\begin{table}[]
		\centering
		\small
		\label{Table_1_OfParameters}
		\caption{List of simulation parameters}
		\begin{tabular}{|l|l|}
			\hline
			\multicolumn{1}{|c|}{\textbf{Parameter}} & \multicolumn{1}{c|}{\textbf{Value}}                         \\ \hline
			Cell Attenuation Model                   & $128+37.1\log(R)$ dB \cite{Moha1812:Multi}  \\ \hline
			Node Bandwidth $(W)$                     & 5 MHz                                                       \\ \hline
			Device proximity $(R)$                   & 500m                                                        \\ \hline
			Transmission Power $(P_k)$               & 23 dBm                                                      \\ \hline
			Noise Power Density $(N_0)$              & -174 dBm/Hz                                                 \\ \hline
			Computation Capabilities $(f_k)$         & $\sim \{~2.4,~1.2\}$ GHz                             \\ \hline
			MNIST Dataset size $(d)$                 & 54,000 images                                               \\ \hline
			MNIST Dataset Features $(\mathcal{F})$   & 784 ($~28 \times 28~$) pixels                               \\ \hline
		\end{tabular}
	\end{table}
	
	For the simulation, we consider a set of $K=20$ learners and test for total training times of $T=\{300,~400,~500,~600\}$s. It was found that a value of $\eta=0.01$ for the learning rate works very well and setting $B_0$ in the range $0.005-0.01$ provided solutions that converge. $\tau_{max}$ is set to the case where only 3 global aggregations would be done on $d_k = d/K ~ \forall k \in \mathcal{K}$.
	\begin{figure}[t!]
		\centering
		\includegraphics[scale=0.5]{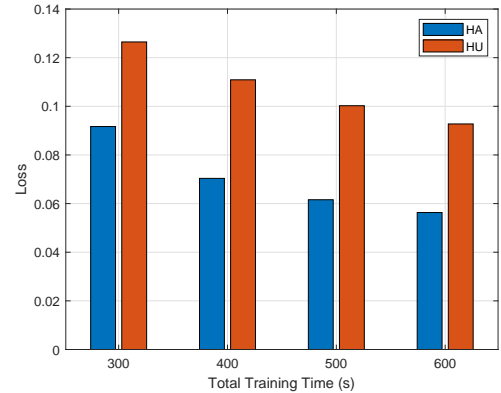}
		\caption{Training loss versus total training time for $K=20$ learners}
		\label{figure3label}
	\end{figure}
	
	\subsubsection{Loss and Validation Accuracy}
	We plot the final loss value after training for time $T$ in figure \ref{figure3label} and the final accuracy in figure \ref{figure4label} for both approaches, the proposed HA approach and the HU approach in \cite{Tuor_01_AdaptiveControl}. As expected, as the training time is increased, the loss value decreases for all approaches. However, for the HA approach, there is just a slight increase in validation accuracy because it is able to achieve that in minimum time. The main conclusion is that optimizing $\tau$ and $d_k$ jointly influences the possible number of global aggregations $G$ and the total iterations $L$ which helps in converging to a lower loss and a higher final accuracy. For example, the loss of the HA approach is lower by 0.03-0.05 which represents gains in the range of 27\% - 40\%. Furthermore, training for 300s using the HA approach achieves a 97\% accuracy within 300s of training, a value not achieved by the HU approach even in 600s.  	
	\begin{figure}[t!]
		\centering
		\includegraphics[scale=0.5]{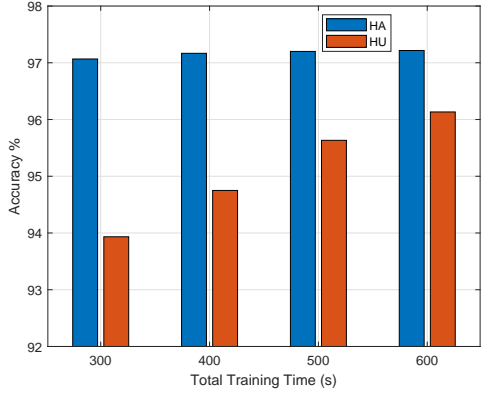}
		\caption{Validation accuracy versus total training time for $K=20$ learner }
		\label{figure4label}
	\end{figure}

	\section{Conclusion}
	This paper extends the efforts towards the MEL paradigm by jointly optimizing the task size allocation for each learner and the number of local ML iterations in a global cycle for distributed ML over the wireless edge. The problem uses existing bounds on the DL paradigm to relate the optimization variables to the loss function which is shown to be convex. It is shown that optimal value of the local updates minimizes the upper bound on the loss difference and the total iterations and batch sizes after every global step. A heuristic approach is proposed to carry out the global learning process. Through simulations, it shown that our HA scheme performed much better in terms of the possible number of updates and learning accuracy compared to the HU scheme. 
	
	\bibliographystyle{IEEEtran}
	\bibliography{Dissertation}

\begin{thebibliography}{10}
\providecommand{\url}[1]{#1}
\csname url@samestyle\endcsname
\providecommand{\newblock}{\relax}
\providecommand{\bibinfo}[2]{#2}
\providecommand{\BIBentrySTDinterwordspacing}{\spaceskip=0pt\relax}
\providecommand{\BIBentryALTinterwordstretchfactor}{4}
\providecommand{\BIBentryALTinterwordspacing}{\spaceskip=\fontdimen2\font plus
\BIBentryALTinterwordstretchfactor\fontdimen3\font minus
  \fontdimen4\font\relax}
\providecommand{\BIBforeignlanguage}[2]{{%
\expandafter\ifx\csname l@#1\endcsname\relax
\typeout{** WARNING: IEEEtran.bst: No hyphenation pattern has been}%
\typeout{** loaded for the language `#1'. Using the pattern for}%
\typeout{** the default language instead.}%
\else
\language=\csname l@#1\endcsname
\fi
#2}}
\providecommand{\BIBdecl}{\relax}
\BIBdecl

\bibitem{N202003_vXchnge_IoT_data}
\BIBentryALTinterwordspacing
K.~Gyarmathy, ``{Comprehensive Guide to IoT Statistics You Need to Know in
  2020},'' 2020. [Online]. Available:
  \url{https://www.vxchnge.com/blog/iot-statistics}
\BIBentrySTDinterwordspacing

\bibitem{Ref1_Original}
\BIBentryALTinterwordspacing
M.~Chiang and T.~Zhang, ``{Fog and IoT: An Overview of Research
  Opportunities},'' \emph{IEEE Internet of Things Journal}, vol.~3, no.~6, pp.
  854--864, dec 2016. [Online]. Available:
  \url{http://ieeexplore.ieee.org/document/7498684/}
\BIBentrySTDinterwordspacing

\bibitem{Ref2_Original}
\BIBentryALTinterwordspacing
{Rhea Kelly}, ``{Internet of Things Data To Top 1.6 Zettabytes by 2020},''
  2015. [Online]. Available:
  \url{https://campustechnology.com/articles/2015/04/15/internet-of-things-data-to-top-1-6-zettabytes-by-2020.aspx}
\BIBentrySTDinterwordspacing

\bibitem{HarvardDistributed}
S.~Teerapittayanon, B.~McDanel, and H.~T. Kung, ``{Distributed Deep Neural
  Networks over the Cloud, the Edge and End Devices},'' \emph{Proceedings -
  International Conference on Distributed Computing Systems}, pp. 328--339,
  2017.

\bibitem{Tuor_01_AdaptiveControl}
\BIBentryALTinterwordspacing
S.~Wang, T.~Tuor, T.~Salonidis, K.~K. Leung, C.~Makaya, T.~He, and K.~Chan,
  ``{When Edge Meets Learning : Adaptive Control for Resource-Constrained
  Distributed Machine Learning},'' in \emph{INFOCOM}, 2018. [Online].
  Available: \url{https://ieeexplore.ieee.org/document/8486403}
\BIBentrySTDinterwordspacing

\bibitem{Wang2019}
\BIBentryALTinterwordspacing
------, ``{Adaptive Federated Learning in Resource Constrained Edge Computing
  Systems},'' \emph{IEEE Journal on Selected Areas in Communications}, no.
  Early Access, pp. 1--1, 2019. [Online]. Available:
  \url{https://ieeexplore.ieee.org/document/8664630/}
\BIBentrySTDinterwordspacing

\bibitem{Mohammad2019a}
\BIBentryALTinterwordspacing
U.~Mohammad and S.~Sorour, ``{Adaptive Task Allocation for Mobile Edge
  Learning},'' in \emph{2019 IEEE Wireless Communications and Networking
  Conference Workshop (WCNCW)}.\hskip 1em plus 0.5em minus 0.4em\relax IEEE,
  apr 2019, pp. 1--6. [Online]. Available:
  \url{https://ieeexplore.ieee.org/document/8902527/}
\BIBentrySTDinterwordspacing

\bibitem{Tuor04}
\BIBentryALTinterwordspacing
D.~Conway-Jones, T.~Tuor, S.~Wang, and K.~K. Leung, ``{Demonstration of
  Federated Learning in a Resource-Constrained Networked Environment},'' in
  \emph{2019 IEEE International Conference on Smart Computing (SMARTCOMP)},
  2019. [Online]. Available:
  \url{https://ieeexplore.ieee.org/abstract/document/8784064}
\BIBentrySTDinterwordspacing

\bibitem{J1902_ResAllocMobDL_TPDS_WuZhang}
\BIBentryALTinterwordspacing
C.~Wu, L.~Zhang, Q.~Li, Z.~Fu, W.~Zhu, and Y.~Zhang, ``{Enabling Flexible
  Resource Allocation in Mobile Deep Learning Systems},'' \emph{IEEE
  Transactions on Parallel and Distributed Systems}, vol.~30, no.~2, pp.
  346--360, feb 2019. [Online]. Available:
  \url{https://ieeexplore.ieee.org/document/8434315/}
\BIBentrySTDinterwordspacing

\bibitem{2019arXiv190907972C}
\BIBentryALTinterwordspacing
M.~Chen, Z.~Yang, W.~Saad, C.~Yin, H.~V. Poor, and S.~Cui, ``{A Joint Learning
  and Communications Framework for Federated Learning over Wireless
  Networks},'' \emph{arXiv e-prints}, p. arXiv:1909.07972, sep 2019. [Online].
  Available: \url{https://arxiv.org/abs/1909.07972}
\BIBentrySTDinterwordspacing

\bibitem{2019arXiv191102417Y}
\BIBentryALTinterwordspacing
Z.~Yang, M.~Chen, W.~Saad, C.~S. Hong, and M.~Shikh-Bahaei, ``{Energy Efficient
  Federated Learning Over Wireless Communication Networks},'' \emph{arXiv
  e-prints}, p. arXiv:1911.02417, nov 2019. [Online]. Available:
  \url{https://arxiv.org/abs/1911.02417v1}
\BIBentrySTDinterwordspacing

\bibitem{MNIST_IEEE}
\BIBentryALTinterwordspacing
Y.~LeCun, L.~Bottou, Y.~Bengio, and P.~Haffner, ``{Gradient-based Learning
  Applied to Document Recognition},'' \emph{Proceedings of IEEE}, vol.~86,
  no.~11, pp. 2278 -- 2324, 1998. [Online]. Available:
  \url{https://ieeexplore.ieee.org/document/726791}
\BIBentrySTDinterwordspacing

\bibitem{Moha1812:Multi}
\BIBentryALTinterwordspacing
U.~Y. Mohammad and S.~Sorour, ``{Multi-Objective Resource Optimization for
  Hierarchical Mobile Edge Computing},'' in \emph{2018 IEEE Global
  Communications Conference: Mobile and Wireless Networks (Globecom2018 MWN)},
  Abu Dhabi, United Arab Emirates, dec 2018, pp. 1--6. [Online]. Available:
  \url{https://ieeexplore.ieee.org/document/8648109}
\BIBentrySTDinterwordspacing

\end{thebibliography}
	
	\section*{Appendix}
	\section*{A. Proof of Theorem \ref{theorem1}}
	\label{appendix}
	The objective function $O$ is strictly convex under certain conditions because $\frac{\partial^2 O}{\partial \tau^2}>0$. Because $\tau$ is a positive integer, the optimal value $\tau^*$ is the argument that minimizes $O(\tau)$. 	
	The reciprocal of $L(\tau)$ in (\ref{eq_O_tau}) can be separated into two terms $M(\tau)$ and $N(\tau)$ as follows:
	\begin{subequations}
		\begin{align}
		M(\tau) =  \dfrac{d}{KT} \dfrac{\prod_{k=1}^{K}(\tau + b_k)}{\sum_{k=1}^{K} \tau \prod_{\substack{l=1 \\ l\neq k}}^{K} \left(\tau^+b_l\right)} \\
		N(\tau) =  \dfrac{1}{KT} \dfrac{\sum_{k=1}^{K} a_k \prod_{\substack{l=1 \\ l\neq k}}^{K} \left(\tau^+b_l\right)}{\sum_{k=1}^{K} \tau \prod_{\substack{l=1 \\ l\neq k}}^{K} \left(\tau^+b_l\right)}
		\end{align}
	\end{subequations}
	Therefore, the objective function can be re-written as $O(\tau)=O_1(\tau)+O_2(\tau)$ where $O_1(\tau) = M(\tau)P(\tau)$ and $O_2(\tau) = N(\tau)P(\tau)$. Moreover, the term $P(\tau)$ can be written as the reciprocal of $\nu(\tau)$ where 
	\begin{equation}
	\nu(\tau) = A-B\frac{C^\tau-1-(C-1)\tau}{\tau}
	\end{equation} 
	The constants $A = \eta\left(1-\frac{\beta\eta}{2}\right)$, $B = \frac{\delta}{\beta}\frac{\rho}{\omega\epsilon^2}$ and $C = \eta\beta+1$. We can say that $B=\frac{\delta}{\beta} B_0$ where $B_0 = \frac{\rho}{\omega\epsilon^2} > 0$ is a control parameter that can be set empirically.
	
	For brevity, we will represent $f(\tau)$ as $f$ where $f$ may be $O$, $O_1$, $O_2$, $M$, $N$ or $P$. We will also represent $\frac{\partial f}{\partial \tau}$ and $\frac{\partial^2 f}{\partial \tau^2}$ as $f^\prime$ and $f^{\prime \prime}$, respectively. Using this new notation, $O^{\prime \prime}$ can be given as follows:
	\begin{equation}
	\label{eq_O_SecondDreivative}
	O^{\prime \prime} = P(M^{\prime \prime}+N^{\prime \prime}) + 2P^\prime(M^\prime+N^\prime) + (M+N)P^{\prime \prime}
	\end{equation}
	
	By definition $M > 0$ and $N > 0$ because they are related to the time consumed by user $k$, and $P > 0$ assuming that the constraint in (\ref{relaxed-minloss-const}) holds. Hence, if we can show that each of the three terms in (\ref{eq_O_SecondDreivative}) are strictly positive, then $O$ will be strictly convex. We need to show that  $P^{\prime \prime} > 0$ and $M^{\prime \prime}+N^{\prime \prime} > 0$. Furthermore, if we can show that $M^\prime+N^\prime$ and $P^\prime$ follow the same sign, $O^{\prime \prime} > 0$. 
	
	The first derivatives of $M$, $N$, and $P$ can be given by:
	\begin{subequations}
		\begin{align}
		M^{\prime} = -\dfrac{d}{KT}\dfrac{\sum_{k=1}^{K}\frac{b_k}{(\tau + b_k)^2}}{\left(\sum_{k=1}^{K}\frac{\tau}{(\tau + b_k)}\right)^2}
		\end{align}
		\begin{align}
		N^{\prime} = \dfrac{M^\prime}{d} -\dfrac{1}{KT} \dfrac{2\sum_{k=1}^{K}\frac{b_k}{(\tau + b_k)^3}}{\left(\sum_{k=1}^{K}\frac{\tau}{(\tau + b_k)}\right)^2}
		\end{align}
		\begin{align}
		P^{\prime} = -B\dfrac{C^\tau[1-(\ln C)\tau] - 1}{\left[A\tau - B(C^\tau-1-(C-1)\tau) \right]^2 } 
		\label{eq_P_prime}
		\end{align}
	\end{subequations}  
	The variables $a_k = \frac{C_k^1}{C_k^2}$ and $b_k = \frac{C_k^0}{C_k^2}$, respectively, and both are positive quantities. For $M^\prime$, the first term outside the square bracket is always negative whereas the term inside is a sum of positive quantities whereas $N^\prime$ is a sum of negative quantities. Therefore, $M^\prime < 0$ and $N^\prime < 0$. Hence, we need to show that $P^\prime < 0$ or find the domain on which $P^\prime < 0$.	
	
	The complete expressions of $M^{\prime \prime}$ and $N^{\prime \prime}$ can be given by:
	\begin{subequations}
		\begin{multline}
		M^{\prime \prime} = \dfrac{d}{KT} \dfrac{1}{\left(\sum_{k=1}^{K}\frac{\tau}{(\tau + b_k)}\right)^2} \\ \left[2\sum_{k=1}^{K}\frac{b_k}{(\tau + b_k)^3} + \sum_{k=1}^{K}\frac{b_k}{(\tau + b_k)^2}\right]
		\end{multline}
		\begin{multline}
		N^{\prime \prime} = \dfrac{1}{KT} \sum_{k=1}^{K} \dfrac{a_k}{\left(\tau+b_k\right)^3 \left(\sum_{l=1}^{K}\frac{\tau}{(\tau + b_l)}\right)^4} \\ \left[\sum_{l=1}^{K} \dfrac{\tau}{\tau+b_l} + 2(\tau + b_k)\sum_{l=1}^{K} \dfrac{1}{(\tau+b_l)^2} + \sum_{l=1}^{K} \dfrac{2}{(\tau+b_l)^3}\right] 
		\end{multline}
	\end{subequations} 
	It can be shown that $M^{\prime \prime}$ and $N^{\prime \prime}$ are strictly greater than zero because they are both a sums of positive terms. Hence, the left-most term in (\ref{eq_O_SecondDreivative}) is strictly positive. Thus we need to show that $P^{\prime\prime}>0$ or find the domain for which this is true. 
	
	Although the complete expression for $P^{\prime \prime}$ is omitted for brevity, it can be shown that the necessary and sufficient conditions to achieve $P^\prime < 0$ and $P^{\prime\prime}>0$ is to satisfy $C^\tau[1-(\ln C)\tau] > 1$. From the Bernoulli inequality, we know that $C^\tau \geq (C-1)\tau+1$. Assuming the worst case where the equality holds, the expression can be written as $[(C-1)\tau+1][1-(\ln C)\tau] > 1$. By expanding the expression, we can show that we need to check the following inequality:
	\begin{equation}
	\tau\left[(\ln C) \tau - C (\ln C) -1 + C\right] > 0
	\end{equation}
	We know that as long as a feasible $\tau^* > 0$ is found, we need to satisfy the second term enclosed by the square brackets. Writing $\tau$ as a function of $C$ and, we can see that the condition on $\tau$ is the following:  
	\begin{equation}
	\tau > \dfrac{C (\ln C) +1 - C}{\ln C} 
	\end{equation}
	Recall that $C= \eta\beta+1$ where $\eta$ is chosen such that $\eta\beta \leq 1$, and $\eta, \beta>0$. Hence, it follows that $1 \leq \eta\beta+1 \leq 2$. If we plot, $f(C) = \frac{C (\ln C) +1 - C}{\ln C} $ against the domain of $C$, we notice that for $\nu^{\prime\prime} < 0$ and hence, for $P^{\prime}$ to be strictly negative and $P^{\prime\prime}$ to be strictly positive, it is sufficient for $\tau > 0$. Hence, we have now proved that $O(\tau)$ is strictly convex because $\frac{\partial^2 O}{\partial \tau^2} > 0$ as long as $\tau$ is a positive integer and the ML model variables are selected as defined. 
	
\end{document}